\documentclass{amsart}
\usepackage{amsmath}
\usepackage{amsfonts}
\usepackage{amsthm, upref}
\usepackage{graphicx}
\usepackage{enumerate}
\usepackage[usenames, dvipsnames]{color}
\usepackage{url}
\usepackage{tikz}
\usepackage{array}

\usetikzlibrary{decorations.pathreplacing}
\input xy
\xyoption{all}

\newcommand{\comment}[1]{}

\newcommand{\dd}{{\mathrm{d}}}
\newcommand{\supp}{{\mathrm{supp}}}

\DeclareMathOperator*{\argmax}{arg\,max}

\begin{document}

\theoremstyle{plain}
\newtheorem{theorem}{Theorem}[section]
\newtheorem{lemma}[theorem]{Lemma}
\newtheorem{proposition}[theorem]{Proposition}
\newtheorem{corollary}[theorem]{Corollary}

\theoremstyle{definition}
\newtheorem{definition}[theorem]{Definition}
\newtheorem{asmp}[theorem]{Assumption}
\newtheorem{notn}[theorem]{Notation}
\newtheorem{problem}[theorem]{Problem}

\theoremstyle{remark}
\newtheorem{remark}[theorem]{Remark}
\newtheorem{example}[theorem]{Example}
\newtheorem{clm}[theorem]{Claim}
\newtheorem{assumption}[theorem]{Assumption}

\numberwithin{equation}{section}

\title[Universal portfolios in stochastic portfolio theory]{Universal portfolios\\in stochastic portfolio theory}

\author{Ting-Kam Leonard Wong}
\address{Department of Mathematics\\University of Southern California\\Los Angeles, CA 90089}
\email{tkleonardwong@gmail.com}


\date{\today}
\begin{abstract}
Consider a family of portfolio strategies with the aim of achieving the asymptotic growth rate of the best one. The idea behind Cover's universal portfolio is to build a wealth-weighted average which can be viewed as a buy-and-hold portfolio of portfolios. When an optimal portfolio exists, the wealth-weighted average converges to it by concentration of wealth. Working under a discrete time and pathwise setup, we show under suitable conditions that the distribution of wealth in the family satisfies a pathwise large deviation principle as time tends to infinity. Our main result extends Cover's portfolio to the nonparametric family of functionally generated portfolios in stochastic portfolio theory and establishes its asymptotic universality.
\end{abstract}

\maketitle

\section{Introduction} \label{sec:intro}
The problem of portfolio selection is to decide, at each point in time, the distribution of capital over the available assets in order to maximize future wealth. For portfolios without short sales, the distribution at time $t$ is given by a {\it portfolio vector} $\pi(t) = \left(\pi_1(t), \ldots, \pi_n(t)\right)$ whose components are non-negative and sum to $1$ (here $n \geq 2$ is the number of assets). Since Markowitz's seminal paper \cite{M52} there has been an explosive growth of literature on the theory and practice of portfolio selection. The mainstream approach, due to Markowitz, consists of two major steps. First we build and estimate a statistical model of the joint distribution of future asset returns (usually specified in terms of the first and second moments). Then, based on the investor's preference and risk aversion (described by a utility function), we compute the optimal portfolio weights. We refer the reader to \cite{CK06} for mathematical details as well as practical considerations.

The above approach depends on the investor's (unobservable) preferences and requires forecasts of returns and risks.  From the point of view of an investment firm which manages a strategy for many investors coming in and out, the classical consumption-based utility may not be appropriate. In the end, performance is what portfolio managers care most about. Moreover, it is well known that the optimal portfolio is highly sensitive to model (mis)specifications and estimation errors (see for example \cite{M89}, \cite{C93} and \cite{DGU09}). Can we construct good portfolios without assuming specific models of preferences and asset prices? In recent years two {\it model-free} approaches emerged which attempt to achieve this goal.

\subsection{Stochastic portfolio theory} \label{sec:spt}
Stochastic portfolio theory, first developed by Fernholz \cite{F02} and extended by Fernholz and Karatzas \cite{FKsurvey} and others, is a descriptive theory of equity market and portfolio selection. Instead of modeling preferences and market equilibrium, the theory constructs and analyzes portfolios using properties of observable market quantities. A major result is the existence of portfolio strategies (called {\it relative arbitrages}) that outperform the market portfolio under suitable conditions.

To explain this more precisely let us introduce some notations. In an equity market with $n$ stocks, let $X_i(t) > 0$ be the market capitalization of stock $i$ at time $t$. The {\it market weight} of stock $i$ is the ratio
\begin{equation} \label{eqn:marketweight}
\mu_i(t) = \frac{X_i(t)}{X_1(t) + \cdots + X_n(t)}.
\end{equation}
The market weights are the portfolio weights of the {\it market portfolio}. It is (under idealized assumptions) a buy-and-hold portfolio representing the overall performance of the market. Suppose we arrange the market weights in descending order:
\begin{equation} \label{eqn:capitaldistribution}
\mu_{(1)}(t) \geq \cdots \geq \mu_{(n)}(t).
\end{equation}
Here the $\mu_{(k)}(t)$'s are the reverse order statistics, and the vector $\left(\mu_{(1)}(t), \ldots, \mu_{(n)}(t)\right)$ of ranked market weights is called the {\it capital distribution} of the market. It was observed (see \cite[Chapter 4]{F02}) that despite price and economic fluctuations, the distribution of capital exhibits remarkable stability over long periods. In particular, the equity market has remained {\it diverse}: the maximum market weight $\max_{1 \leq i \leq n} \mu_i(t)$ has been bounded away from $1$. Moreover, the market appears to possess {\it sufficient volatility}: if one plots the cumulative realized volatility of $\mu(t)$, its slope is bounded below. If we assume that the market is diverse and sufficiently volatile,  trading is frictionless and the investor does not influence prices,  there exist portfolios that are guaranteed to outperform the market portfolio over sufficiently long horizons. For precise statements and their relationship with the classical notion of arbitrage, see \cite[Chapter 2]{FKsurvey}. Also see \cite{P16, FKR16} and their references for results concerning short term relative arbitrage. These relative arbitrages are constructed using {\it functionally generated portfolios} which are explicit deterministic functions of the current market weights given by gradients of concave functions. In \cite{PW14} and \cite{PW16} we established an elegant connection between functionally generated portfolio, convex analysis, optimal transport and information geometry. Intuitively, these portfolios work by capturing market volatility, and we showed in \cite{PW14} that functionally generated portfolios exhaust the class of volatility harvesting portfolio maps.

\subsection{Universal portfolio theory}
Universal portfolio theory is a very active field in mathematical finance and machine learning. Instead of giving an extensive review (which we refer the reader to the recent survey \cite{LH14}), let us explain the main ideas of Cover's classic paper \cite{C91} which started the subject. A portfolio of $n$ stocks is said to be {\it constant-weighted}, or {\it constantly rebalanced}, if the portfolio weights $\pi(t) \equiv \pi$ are constant over time. It has been observed empirically that a rebalanced portfolio frequently outperforms a buy-and-hold portfolio of the constituent stocks (see \cite{PW13} for a theoretical justification). Let $Z_{\pi}(t)$ be the wealth of the constant-weighted portfolio $\pi$ at time $t$ (with initial value $Z_{\pi}(0) = 1$), where $\pi$ ranges over the closed unit simplex
\[
\overline{\Delta}_n = \left\{p = (p_1, \ldots, p_n) \in [0, 1]^n: \sum_{i = 1}^n p_i = 1\right\}.
\]
Working with a discrete time market model, Cover asked the following question: Without any knowledge of future stock prices, is it possible to invest in such a way that the resulting wealth is close to
\[
Z^*(t) = \max_{\pi \in \overline{\Delta}_n} Z_{\pi}(t),
\]
the performance of the best constant-weighted portfolio chosen with hindsight? While this seems to be an unrealistically ambitious goal, Cover constructed a non-anticipative sequence of portfolio weights $\widehat{\pi}(t)$ such that the resulting wealth $\widehat{Z}(t)$ satisfies the {\it universality property}
\begin{equation} \label{eqn:universalityCover}
\frac{1}{t} \log \frac{\widehat{Z}(t)}{Z^*(t)} \geq \frac{C}{t^{(n-1)/2}} \rightarrow 0,
\end{equation}
where $C > 0$ is a constant, for {\it arbitrary} sequences of stock returns. Explicitly, Cover's {universal portfolio} is given by
\begin{equation}  \label{eqn:coverportfolio}
\widehat{\pi}(t) = \frac{\int_{\overline{\Delta}_n} \pi Z_{\pi}(t) d\pi}{\int_{\overline{\Delta}_n} Z_{\pi}(t) d\pi}.
\end{equation}
That is, $\widehat{\pi}(t)$ is the average of all constant-weighted portfolios weighted by their performances. In fact, it can be shown that
\begin{equation} \label{eqn:coverportfolio.identity}
\widehat{Z}(t) = \frac{\int_{\overline{\Delta}_n} Z_{\pi}(t) d\pi}{\int_{\overline{\Delta}_n} d\pi}.
\end{equation}
The representation \eqref{eqn:coverportfolio.identity} allows us to view Cover's portfolio as a buy-and-hold portfolio of all constant-weighted portfolios, where each portfolio receives the same infinitesimal wealth initially. Cover's result \eqref{eqn:universalityCover} states that the maximum and average of $V_{\pi}(t)$ over $\pi \in \overline{\Delta}_n$ have the same asymptotic growth rate, and can be viewed as a consequence of Laplace's method of integration and the fact that for constant-weighted portfolios the map $\pi \mapsto V_{\pi}(t)$ is essentially a multiple of a Gaussian density. While numerous alternative portfolio selection algorithms have been proposed for constant-weighted and other families of portfolios, the idea of forming a wealth-weighted average underlies many of these generalizations.

\subsection{Summary of main results} \label{sec:summary}
It is natural to ask whether functionally generated portfolios and Cover's universal portfolio are connected in some way (see \cite[Remark 11.7]{FKsurvey}). Recently, \cite{B14} showed that Cover's portfolio \eqref{eqn:coverportfolio} is functionally generated in a generalized sense. With hindsight, this result is not surprising since Cover's portfolio is a buy-and-hold portfolio of constant-weighted portfolios, and both buy-and-hold and constant-weighted portfolios are functionally generated \cite[Example 3.1.6]{F02}. Instead, it is more interesting to think of Cover's portfolio as a market portfolio where each constituent asset is the value process of a portfolio in a family. The capital distribution \eqref{eqn:capitaldistribution} then generalizes to the {\it distribution of wealth} over the portfolios, a measure-valued process. While the capital distribution of an equity market is typically stable and diverse, this is not true for the distribution of wealth over a typical family of portfolios. Quite the contrary, wealth often concentrates {\it exponentially} around an optimal portfolio, and under suitable conditions this can be quantified by a pathwise large deviation principle (LDP). Moreover, we show that Cover's portfolio \eqref{eqn:universalityCover} can be generalized to the nonparametric family of functionally generated portfolios which contains the constant-weighted portfolios.

In this paper we study the long term performance of various portfolios. To state the main results let us introduce informally some concepts. These as well as the assumptions will be stated precisely in Section \ref{sec:prelim}. We consider an idealized equity market with $n \geq 2$ non-dividend paying stocks in discrete time ($t = 0, 1, 2, \ldots$). The evolution of the market is modeled by a sequence $\{\mu(t)\}_{t = 0}^{\infty}$ of market weights with values in the open unit simplex $\Delta_n$. For technical reasons, we follow \cite{C91} and assume that there is a constants $M > 0$ such that $\frac{1}{M} \leq \frac{\mu_i(t + 1)}{\mu_i(t)} \leq M$ for all $i$ and $t$ ($M$ is unknown to the investor). Consider a family $\{\pi_{\theta}\}_{\theta \in \Theta}$ of portfolio maps, where $\Theta$ is a topological index set and each $\pi_{\theta}$ is a map from $\Delta_n$ to $\overline{\Delta}_n$. If the investor chooses the portfolio map $\pi_{\theta}$, the portfolio weight vector at time $t$ is given by $\pi_{\theta}(\mu(t))$ which depends only on $\mu(t)$. For convenience and following the tradition of stochastic portfolio theory, we measure the values of all portfolios relative to that of the market portfolio. Thus we define the {\it relative value} $V_{\theta}(t)$ of the self-financing portfolio $\pi_{\theta}$ by
\begin{equation} 
V_{\theta}(0) = 1, \quad V_{\theta}(t + 1) = V_{\theta}(t) \sum_{i = 1}^n \pi_{\theta, i}(\mu(t)) \frac{\mu_i(t + 1)}{\mu_i(t)}.
\end{equation}
(See Definition \ref{def:V}). Imagine at time $0$ we distribute wealth over the family according to a Borel probability measure $\nu_0$ on $\Theta$; we call $\nu_0$ the {\it initial distribution}. The {\it wealth distribution} of the family $\{\pi_{\theta}\}_{\theta \in \Theta}$ at time $t$ is the Borel probability measure $\nu_t$ on $\Theta$ defined by
\begin{equation} \label{eqn:wealthdistribution0}
\nu_t(B) = \frac{1}{\int_{\Theta} V_{\theta}(t) d\nu_0(\theta)} \int_B V_{\theta}(t) d\nu_0(\theta), \quad B \subset \Theta.
\end{equation}

We are interested in situations where the wealth distribution of the family $\{\pi_{\theta}\}_{\theta \in \Theta}$ concentrates exponentially around some optimal portfolio. A natural way to quantify this is to prove a {\it large deviation principle} (LDP). A standard reference of large deviation theory is \cite{DZ98}.

\begin{definition} \label{def:LDP}
Let $I: \Theta \rightarrow [0, \infty]$ be a lower-semicontinuous function, called the rate function. We say that the sequence $\{\nu_t\}_{t = 0}^{\infty}$ satisfies the large deviation principle on $\Theta$ with rate $I$ if the following statements hold.
\begin{enumerate}
\item[(i)] (Upper bound) For every closed set $F \subset \Theta$,
\[
\limsup_{t \rightarrow \infty} \frac{1}{t} \log \nu_t(F) \leq -\inf_{\theta \in F} I(\theta).
\]
\item[(ii)] (Lower bound) For every open set $G \subset \Theta$,
\[
\liminf_{t \rightarrow \infty} \frac{1}{t} \log \nu_t(G) \geq -\inf_{\theta \in G} I(\theta).
\]
\end{enumerate}
\end{definition}

A sufficient condition for existence of LDP is that the {\it asymptotic growth rate}
\begin{equation} \label{eqn:arate}
W(\theta) = \lim_{t \rightarrow \infty} \frac{1}{t} \log V_{\theta}(t)
\end{equation}
exists for all $\theta \in \Theta$ and the map $\theta \mapsto V_{\theta}(t)$ is `sufficiently regular'. As preparation, in Section \ref{sec:discrete} we study a simple situation where the family $\{\pi_{\theta}\}_{\theta \in \Theta}$, as maps from $\Delta_n$ to $\overline{\Delta}_n$, is {\it totally bounded} in the uniform metric.

\begin{theorem} \label{thm:main1}
Let  $\{\pi_{\theta}\}_{\theta \in \Theta}$ be a totally bounded family of portfolio maps from $\Delta_n$ to $\overline{\Delta}_n$. Suppose the asymptotic growth rate $W(\theta) = \lim_{t \rightarrow \infty} \frac{1}{t} \log V_{\theta}(t)$ exists for all $\theta \in \Theta$ and the initial distribution $\nu_0$ has full support on $\Theta$. Then the sequence $\nu_t$ of wealth distributions satisfies LDP on $\Theta$ with rate function
\[
I(\theta) = W^* - W(\theta),
\]
where $W^* = \sup_{\theta \in \Theta} W(\theta)$.
\end{theorem}

In Section \ref{sec:fgp} we consider the family of functionally generated portfolios in stochastic portfolio theory. Following \cite{PW14} and \cite{W14}, we say that a portfolio map $\pi: \Delta_n \rightarrow \overline{\Delta}_n$ is {\it functionally generated} if there exists a concave function $\Phi: \Delta_n \rightarrow (0, \infty)$ such that
\begin{equation} \label{eqn:fgpineq}
\sum_{i = 1}^n \pi_i(p) \frac{q_i}{p_i} \geq \frac{\Phi(q)}{\Phi(p)}
\end{equation}
for all $p, q \in \Delta_n$. The function $\Phi$ is called the {\it generating function} of $\pi$. Geometrically, \eqref{eqn:fgpineq} means that the vector $\left(\frac{\pi_1(p)}{p_1}, \ldots, \frac{\pi_n(p)}{p_n}\right)$ defines a supergradient of the (exponentially) concave function $\varphi = \log \Phi$ at $p$. Conversely, any positive concave function on $\Delta_n$ generates a functionally generated portfolio. As an example, the constant-weighted portfolio $\left(\pi_1, \ldots, \pi_n\right)$ where $\pi \in \overline{\Delta}_n$ is generated by the geometric mean $\Phi(p) = p_1^{\pi_1} \cdots p_n^{\pi_n}$. We denote the family of functionally generated portfolios by ${\mathcal{FG}}$. We endow ${\mathcal{FG}}$, as a space of functions from $\Delta_n$ to $\overline{\Delta}_n$ with the topology of uniform convergence. It is clear that ${\mathcal{FG}}$ is infinite dimensional and is thus `nonparametric'. Nevertheless, it can be shown that ${\mathcal{FG}}$ is convex.

Given a market path $\{\mu(t)\}_{t = 0}^{\infty} \subset \Delta_n$, let
\begin{equation} \label{eqn:empirical.measure}
{\Bbb P}_t = \frac{1}{t} \sum_{s = 0}^{t - 1} \delta_{(\mu(s), \mu(s + 1))}
\end{equation}
be the empirical measure of the pair $\left(\mu(s), \mu(s + 1)\right)$ up to time $t$. We have mentioned in Section \ref{sec:spt} that the capital distribution of the market is stable in the long run. Mathematical modeling of this stability led to active development in rank-based diffusion processes (see for example \cite{BFK05} and \cite{IPBKF11}). In our context, it seems natural to impose an asymptotic condition on the sequence $\{{\Bbb P}_t\}_{t = 0}^{\infty}$ in the spirit of \cite{J92}. The following is the main result of this paper.

\begin{theorem} \label{thm:main2}
Suppose ${\Bbb P}_t$ converges weakly to an absolutely continuous Borel probability measure ${\Bbb P}$ on $\Delta_n \times \Delta_n$.
\begin{itemize}
\item[(i)] (Glivenko-Cantelli property) The asymptotic growth rate $W(\pi)$ defined by \eqref{eqn:arate} exists for all $\pi \in {\mathcal{FG}}$. Furthermore, we have
\[
\lim_{t \rightarrow \infty} \sup_{\pi \in {\mathcal{FG}}} \left| \frac{1}{t} \log V_{\pi}(t) - W(\pi) \right| = 0.
\]
\item[(ii)] (LDP) Let $\nu_0$ be any initial distribution on ${\mathcal{FG}}$. Then the sequence $\{\nu_t\}_{t = 0}^{\infty}$ of wealth distributions given by \eqref{eqn:wealthdistribution0} satisfies LDP with rate
\[
I(\pi) =
\begin{cases}
W^* - W(\pi) &\mbox{if } \pi \in \supp(\nu_0), \\
\infty & \mbox{otherwise}, \end{cases}
\]
where $W^* = \sup_{\pi \in \supp(\nu_0)} W(\pi)$.
\item[(iii)] (Universality) There exists a probability distribution $\nu_0$ on ${\mathcal{FG}}$ such that $\sup_{\pi \in \supp(\nu_0)} W(\pi) = W^* := \sup_{\pi \in {\mathcal{FG}}} W(\pi)$ for any absolutely continuous ${\Bbb P}$ (see \eqref{eqn:mynu}). For this initial distribution, consider Cover's portfolio
\begin{equation} \label{eqn:coverfgp}
\widehat{\pi}(t) := \int_{{\mathcal{FG}}} \pi(\mu(t)) d\nu_t(\pi).
\end{equation}
Let $\widehat{V}(t)$ be the relative value of this portfolio and let $V^*(t) = \sup_{\pi \in {\mathcal{FG}}} V_{\pi}(t)$. Then
\begin{equation} \label{eqn:thmuniversality}
\lim_{t \rightarrow \infty} \frac{1}{t} \log \widehat{V}(t) = \lim_{t \rightarrow \infty} \frac{1}{t} \log V^*(t) = W^*.
\end{equation}
In particular, we have $\lim_{t \rightarrow \infty} \frac{1}{t} \log \left(\widehat{V}(t) / V^*(t)\right) = 0$.
\end{itemize}
\end{theorem}

For example, if $\{\mu(t)\}$ is an ergodic time homogeneous Markov chain, we may take ${\Bbb P}$ to be the stationary distribution of $\{(\mu(t), \mu(t + 1))\}$.

In \cite{W14} we studied an optimization problem for functionally generated portfolio analogous to nonparametric density estimation. Regarding $\log V_{\pi}(t)$ as the log likelihood function for estimating $\pi$ and $\nu_0$ as the prior distribution, Theorem \ref{thm:main2}(ii) shows that the posterior distribution $\nu_t$ satisfies an LDP. Convergence properties of posterior distributions in nonparametric statistics are delicate (see for example \cite{BSW99}) and large deviation results are rare. For Dirichlet priors an LDP is proved in \cite{GO00}. Theorem \ref{thm:main2}(iii) shows that the posterior mean \eqref{eqn:coverfgp} performs asymptotically as good as the best portfolio in ${\mathcal{FG}}$. If we think of the results in \cite{W14} as point estimation of functionally generated portfolio by maximum likelihood, Theorem \ref{thm:main2} gives the Bayesian counterpart. 

Another natural question is to relate Cover's universal portfolio with the num\'{e}raire portfolio (also called the log-optimal portfolio). In the context of stochastic portfolio theory, this question is studied in \cite{CSW16} in both discrete and continuous time.

For practical applications we would like to strengthen Theorem \ref{thm:main2} to include quantitative bounds as well as algorithms for computing $\widehat{\pi}$. This and other further problems are gathered in Section \ref{sec:conclusion}.

\section{Wealth distributions of portfolios} \label{sec:prelim}
\subsection{Stock and market weight}
We consider an equity market with $n \geq 2$ non-dividend stocks in discrete time. The dynamics of the market will be specified in terms of the {\it market weights} $\mu(t) = \left(\mu_1(t), \ldots, \mu_n(t)\right)$ given by \eqref{eqn:marketweight}. The vector of market weights $\mu(t)$ takes values in the open unit simplex $\Delta_n$ in $\mathbb{R}^n$. Suppose the market capitalization of stock $i$ at time $t$ is $X_i(t)$ and its simple return over the time interval $[t, t + 1]$ is $R_i(t)$. The market weights at time $t + 1$ are then given by
\[
\mu_i(t + 1) = \frac{X_i(t)(1 + R_i(t))}{X_1(t)(1 + R_1(t)) + \cdots + X_n(t)(1 + R_n(t))}.
\]
We visualize the market as a discrete path in $\Delta_n$. This includes only changes in capitalizations due to returns and excludes implicitly all changes due to corporate actions such as public offerings. We assume that $\{\mu(t)\}_{t = 0}^{\infty}$ is an arbitrary sequence in $\Delta_n$; in particular, no underlying probability space is involved. The assumptions we state will be in terms of the path properties of the sequence $\{\mu(t)\}_{t = 0}^{\infty}$. One such assumption is the following.

\begin{assumption} \label{ass:bound}
There exists a constant $M > 0$ such that the market weight sequence $\{\mu(t)\}_{t = 0}^{\infty}$ satisfies
\begin{equation} \label{eqn:marketassumption}
\frac{1}{M} \leq \frac{\mu_i(t + 1)}{\mu_i(t)} \leq M
\end{equation}
for all $1 \leq i \leq n$ and $t \geq 0$. Let
\begin{equation} \label{eqn:pairstatespace}
{\mathcal{S}} = \left\{ (p, q) \in \Delta_n \times \Delta_n : \frac{1}{M} \leq \frac{q_i}{p_i} \leq M \text{ for } 1 \leq i \leq n\right\}.
\end{equation}
Then \eqref{eqn:marketassumption} states that $\left(\mu(t), \mu(t + 1)\right) \in {\mathcal{S}}$ for all $t \geq 0$.
\end{assumption}

Assumption \ref{ass:bound} states that the relative returns of the stocks are bounded; this is purely for technical reasons and can be found in previous work such as \cite{C91, HSSW98, CB03, HK15}. Note that the value of $M$ is unknown to the investor. While the assumptions that the stocks do not die (since $\mu(t) \in \Delta_n$ for all $t$) and do not pay dividends are unrealistic, they are imposed to reduce technicalities so that we can focus on the key ideas concerning long term properties of portfolios. Using a more general model, one can reinvest dividends and consider varying number of stocks, but this would complicate the analysis. Similar assumptions are common in stochastic portfolio theory (see \cite[Chapter 1]{F02} and \cite[Chapter 1]{FKsurvey}). 

\subsection{Portfolio and relative value}
A {\it portfolio vector} is an element of $\overline{\Delta}_n$, the closed unit simplex. All portfolios considered are fully invested in the stock market, and short selling is prohibited. At each time $t$ the investor chooses a portfolio vector $\pi(t)$, and the performance of the resulting self-financing portfolio will be measured relative to the market portfolio. Formally, we define
\[
V_{\pi}(t) = \frac{\text{growth of \$1 of the portfolio }\pi}{\text{growth of \$1 of the market portfolio }\mu}.
\]

\begin{definition} [Relative value] \label{def:V}
Let $\{\pi(t)\}_{t = 0}^{\infty}$ be sequence of portfolio vectors. Given the market weight sequence $\{\mu(t)\}_{t = 0}^{\infty}$, the relative value of $\pi$ (with respect to the market portfolio) is the sequence $\{V_{\pi}(t)\}_{t = 0}^{\infty}$ defined by $V_{\pi}(0) = 1$ and
\begin{equation} \label{eqn:V}
\frac{V_{\pi}(t + 1)}{V_{\pi}(t)} = \sum_{i = 1}^n \pi_i(t) \frac{\mu_i(t + 1)}{\mu_i(t)} =: \pi(t) \cdot \frac{\mu(t + 1)}{\mu(t)}, \quad t \geq 0.
\end{equation}
Here $a \cdot b$ is the Euclidean inner product and $a / b$ is the vector of componentwise ratios whenever they are well-defined.
\end{definition}

For a derivation of \eqref{eqn:V} see \cite{PW13}. In \eqref{eqn:V}, it is implicitly assumed that the investor is a price taker and the trades do not influence prices. We will restrict to portfolio strategies that are deterministic functions of the current market weight, i.e., $\pi(t) = \pi(\mu(t))$. In this case a portfolio strategy is fully specified by a mapping $\pi: \Delta_n \rightarrow \overline{\Delta}_n$.

\begin{definition} [Portfolio map]
A portfolio map is a mapping $\pi: \Delta_n \rightarrow \overline{\Delta}_n$. The market portfolio is the identity map $\pi(p) = p$ and will be denoted by $\mu$. A portfolio is said to be constant-weighted if $\pi$ is identically constant.
\end{definition}

\subsection{Cover's portfolio as a market portfolio of portfolios}
Let $\Theta$ be an index set and suppose each $\theta \in \Theta$ is associated with a portfolio map $\pi_{\theta}: \Delta_n \rightarrow \overline{\Delta}_n$. The individual components of $\pi_{\theta}$ will be denoted by $\left(\pi_{\theta, 1}, \ldots, \pi_{\theta, n}\right)$. (Sometimes we will use $\pi_1, ..., \pi_k$ to refer to a sequence of portfolios, and the meaning should be clear from the context.) We are interested in the properties of $V_{\theta}(t) := V_{\pi_{\theta}}(t)$ as a function of {\it both} $t$ and $\theta$. To this end, we will consider an imaginary market whose basic assets are the portfolios $\pi_{\theta}$.

We assume that $\Theta$ is a topological space and we are given a Borel probability measure $\nu_0$ on $\Theta$. The measure $\nu_0$ will be called the {\it initial distribution}. The {\it support} $\supp(\nu_0)$ of $\nu_0$ is the smallest closed subset $F$ of $\Theta$ satisfying $\nu_0(F) = 1$. We say that $\nu_0$ has {\it full support} if $\supp(\nu_0) = \Theta$. Intuitively, the imaginary market is defined by distributing unit wealth at time $0$ over the portfolios $\{\pi_{\theta}\}_{\theta \in \Theta}$ according to the initial distribution $\nu_0$, and letting the portfolios evolve. At time $0$, the portfolio $\pi_{\theta}$ receives wealth $\nu_0(d\theta)$ which grows to $V_{\theta}(t) \nu_0(d\theta)$ at time $t$. Thus
\begin{equation} \label{eqn:universalportfolio}
\widehat{V}(t) := \int_{\Theta} V_{\theta}(t) d\nu_0(\theta)
\end{equation}
is the total relative value of the imaginary market at time $t$. In order that \eqref{eqn:universalportfolio} and related quantities (such as \eqref{eqn:wealthdistribution}) are well defined, we assume that the map $(p, \theta) \mapsto \pi_{\theta}(p)$ on $\Delta_n \times \Theta$ is jointly measurable in $(p, \theta)$. Measurability usually follows immediately from the definition of the family considered. By Assumption \ref{ass:bound} and \eqref{eqn:V} we have $V_{\pi}(t + 1) / V_{\pi}(t) \leq M$ for any portfolio, so $V^*(t) < \infty$ and the integral in \eqref{eqn:universalportfolio} is finite.

\begin{definition} [Wealth distribution]
Given a family of portfolios $\{\pi_{\theta}\}_{\theta \in \Theta}$ and an initial distribution $\nu_0$, the wealth distribution is the sequence of Borel probability measures $\{\nu_t\}_{t = 0}^{\infty}$ on $\Theta$ defined by
\begin{equation} \label{eqn:wealthdistribution}
\nu_t(B) = \frac{1}{\widehat{V}(t)} \int_B  V_{\theta}(t) d\nu_0(\theta),
\end{equation}
where $B$ ranges over the measurable subsets of $\Theta$.
\end{definition}

Note that $\frac{d\nu_t}{d\nu_0}(\theta) = \frac{1}{\widehat{V}(t)} V_{\theta}(t)$. The main interest in the quantity $\widehat{V}(t)$ is the following fact first exploited by Cover in \cite{C91} (where $\{\pi_{\theta}\}_{\theta \in \Theta}$ is the family of constant-weighted portfolios). A proof can be found in \cite[Lemma 3.1]{CB03}.

\begin{lemma} [Cover's portfolio]
For each $t$, define the portfolio weight vector
\begin{equation} \label{eqn:posteriormean}
\widehat{\pi}(t) := \int_{\Theta} \pi_{\theta}(\mu(t)) d\nu_t(\theta).
\end{equation}
Then $V_{\widehat{\pi}}(t) \equiv \widehat{V}(t)$ for all $t$. We call $\widehat{\pi}$ Cover's portfolio.
\end{lemma}

For each time $t$, let
\begin{equation} \label{eqn:best}
V^*(t) = \sup_{\theta \in \Theta} V_{\theta}(t)
\end{equation}
be the performance of the best portfolio in the family over the time interval $[0, t]$. The original goal of Cover's portfolio \eqref{eqn:posteriormean} is to track $V^*(t)$ in the sense that
\begin{equation} \label{eqn:universality0}
\frac{1}{t} \log \frac{\widehat{V}(t)}{V^*(t)} \rightarrow 0
\end{equation}
as $t \rightarrow \infty$. If \eqref{eqn:universality0} holds, the portfolio $\widehat{\pi}$ performs asymptotically as good as the best portfolio in the family. In Section \ref{sec:example} we give a simple example to show that \eqref{eqn:universality0} does not always hold. The asymptotic behavior of \eqref{eqn:universality0} naturally links to the concentration of the wealth distribution and motivated our study.

\begin{remark} \label{rem:stat}
As pointed out by several authors (see for example \cite{CB03}),  the construction of Cover's portfolio \eqref{eqn:posteriormean} as a wealth-weighted average has a strong Bayesian flavor. Imagine the problem of finding the best portfolio in the family $\{\pi_{\theta}\}_{\theta \in \Theta}$. Little is known at time $0$, but from historical data, experience and insider knowledge one may form a {\it prior distribution} $\nu_0$ which describes the belief of the investor. At time $t$, having observed the returns of the portfolios up to time $t$, the investor updates the belief with the {\it posterior distribution} $\nu_t$ which satisfies
\[
\frac{\dd \nu_t}{\dd \nu_0}(\theta) \propto \frac{V_{\theta}(t)}{V_{\theta}(0)} = V_{\theta}(t).
\]
This corresponds to Bayes' rule where the relative return plays the role of the likelihood. Note that this procedure is time-consistent. Namely, for $t > s$, we have
\[
\frac{\dd \nu_t}{\dd \nu_s}(\theta) \propto \frac{V_{\theta}(t)}{V_{\theta}(s)}.
\]
Cover's portfolio \eqref{eqn:posteriormean} is then the posterior mean of $\pi_{\theta}(\mu(t))$.
\end{remark}

\section{LDP for totally bounded families} \label{sec:discrete}
To gain intuition about how Cover's portfolio and the wealth distribution behave for a general (possibly nonparametric) family, and to prepare for the more technical treatment of functionally generated portfolio in Section \ref{sec:fgp}, in this section we study large deviation properties of wealth distributions where the family of portfolios is totally bounded with respect to the uniform metric. We will use the following representation of portfolio value which is a direct consequence of Definition \ref{def:V}.

\begin{lemma} \label{lem:integral}
Let $\pi: \Delta_n \rightarrow \overline{\Delta}_n$ be a portfolio map. Then
\begin{equation}
\frac{1}{t} \log V_{\pi}(t) = \int_{\Delta_n \times \Delta_n} \ell_{\pi}(p, q) \dd {\Bbb P}_t(p, q)
\end{equation}
for all $t \geq 1$, where
\begin{equation} \label{eqn:returnfunction}
\ell_{\pi}(p, q) := \log \left( \pi(p) \cdot \frac{q}{p} \right),
\end{equation}
and ${\Bbb P}_t$, defined by \eqref{eqn:empirical.measure}, is the empirical measure of the pair $(\mu(s), \mu(s+1))$ up to time $t$.
\end{lemma}

\subsection{Finite state}
To fix ideas we begin with an even simpler situation where the sequence $\{\mu(t)\}_{t = 0}^{\infty}$ takes values in a {\it finite} set $E \subset \Delta_n$. The finite set $E$ may be obtained by approximating the simplex by a finite grid. Let
\[
\Theta = \left\{\pi: E \rightarrow \overline{\Delta}_n\right\} = \left(\overline{\Delta}_n\right)^E
\]
be the set of all portfolio maps on $E$. (Note that the family is indexed by the symbol $\pi$ itself.) We equip $\Theta$ with the topology of uniform convergence. Since $E$ is finite, this is the same as the topology of pointwise convergence. Note that $\Theta$ is compact and convex. 

\begin{lemma} \label{lem:logoptimal}
Suppose ${\mathbb{P}}_t$ converges weakly to a probability measure ${\mathbb{P}}$ on $E \times E$. Then for each $\pi \in \Theta$, the asymptotic growth rate exists and we have
\[
W(\pi) = \lim_{t \rightarrow \infty} \frac{1}{t} \log V_{\pi}(t) = \int_{E \times E} \ell_{\pi} \dd {\Bbb P},
\]
where $\ell_{\pi}$ is given by \eqref{eqn:returnfunction}. Moreover, there is a portfolio $\pi^* \in \Theta$ satisfying
\begin{equation} \label{eqn:logoptimal}
W(\pi^*) = W^* := \max_{\pi \in \Theta} W(\pi).
\end{equation}
If we write ${\Bbb P}(p, q) = {\Bbb P}_1(p) {\Bbb P}_2\left(q \mid p\right)$, where ${\Bbb P}_1$ is the first marginal and ${\Bbb P}_2$ is the conditional distribution, then
\begin{equation} \label{eqn:logoptimalportfolio}
\pi^*(p) = \argmax_{x \in \overline{\Delta}_n} \int_E \log \left( x \cdot \frac{q}{p}\right) {\Bbb P}_2\left(d q \mid p\right)
\end{equation}
for all $p$ where ${\Bbb P}_1(p) > 0$.
\end{lemma}

A portfolio satisfying \eqref{eqn:logoptimal} may be called a log-optimal portfolio map.

\begin{proof}
Since $E \times E$ is a finite set, by weak convergence we have
\[
W(\pi) = \lim_{t \rightarrow \infty}  \int_{E \times E}  \ell_{\pi} d {\Bbb P}_t = \int_{E \times E}\ell_{\pi} d {\Bbb P}.
\]
Thus the asymptotic growth rate exists for all $\pi \in \Theta$. Clearly $W(\cdot)$ is a continuous function on $\Theta$. Since $\Theta$ is compact, it has a maximizer $\pi^*$. The last statement follows from the representation
\[
W(\pi) = \int_E \left( \int_E \ell_{\pi}(p, q) {\Bbb P}_2\left( \dd q \mid p \right) \right) {\Bbb P}\left(d p\right). \qedhere
\]
\end{proof}

The following LDP is a special case of Theorem \ref{thm:main1} which will be proved in the next subsection.

\begin{theorem} [Finite state LDP] \label{thm:LDPfinite}
Suppose $\{\mu(t)\}_{t = 0}^{\infty}$ takes values in a finite set $E \subset \Delta_n$. Let $\Theta = \left(\overline{\Delta}_n\right)^E$ and suppose that the initial distribution $\nu_0$ has full support.
\begin{enumerate}
\item[(i)] Cover's portfolio $\widehat{\pi}$ defined by \eqref{eqn:posteriormean} satisfies the universality property \eqref{eqn:universality}:
\begin{equation} \label{eqn:universality}
\lim_{t \rightarrow \infty} \frac{1}{t} \log \frac{\widehat{V}(t)}{V^*(t)} = 0.
\end{equation}
\item[(ii)] If ${\Bbb P}_t$ converges weakly to a probability measure ${\Bbb P}$ on $E \times E$, the family $\{\nu_t\}_{t = 0}^{\infty}$ satisfies the large deviation principle on $\Theta$ with the convex rate function
\[
I(\pi) = W^* - W(\pi).
\]
\end{enumerate}
\end{theorem}

\begin{remark}
In the setting of Theorem \ref{thm:LDPfinite}(i), it is not difficult to show (see \cite[Theorem 3.1]{CB03}) that $V^*(t) / \widehat{V}(t)$ is bounded above by a constant multiple of $t^d$, where $d = |E|(n - 1)$ is the `dimension' of $\Theta$ and $|E|$ is the cardinality of $E$.
\end{remark}

\subsection{LDP for totally bounded families} \label{sec:discreteext}
In this subsection we prove Theorem \ref{thm:main1}. Now $\{\mu(t)\}_{t = 0}^{\infty}$ is any sequence in $\Delta_n$ satisfying Assumption \ref{ass:bound}.

Let $\Theta$ be a subset of $L^{\infty}\left(\Delta_n, \overline{\Delta}_n\right)$, the set of functions from $\Delta_n$ to $\overline{\Delta}_n$ equipped with the supremum metric $\|\cdot\|_{\infty}$ (defined in terms of the Euclidean norm $|\cdot|$ on $\overline{\Delta}_n$). We endow $\Theta$ with the induced topology, i.e., the topology of uniform convergence. A consequence of Assumption \ref{ass:bound} is that the function $\ell_{\pi}(\cdot, \cdot)$ defined by \eqref{eqn:returnfunction} is bounded on ${\mathcal{S}}$ between $\log \frac{1}{M}$ and $\log M$, for any $\pi \in \Theta$.

We say that $\Theta$ is {\it totally bounded} if for any $\epsilon > 0$, there exists $\pi_1, \ldots, \pi_N \in \Theta$ with the following property: for any $\pi \in \Theta$, there exists $1 \leq j \leq N$ such that $\|\pi - \pi_j\|_{\infty} < \varepsilon$. The smallest such $N$ is called the $\epsilon$-{\it covering number} of $\Theta$. Thus $\Theta$ is totally bounded if and only if the covering number is finite for all $\epsilon > 0$. For example, if $\Theta = \{\pi(\cdot) \equiv \pi: \pi \in \overline{\Delta}_n\}$ is the family of constant-weighted portfolios, then $\Theta \cong \overline{\Delta}_n$ is compact and hence is totally bounded. Similar ideas are used in \cite{CSW16} where  certain spaces of Lipschitz portfolio maps are studied.

First we prove a lemma which generalizes \cite[Theorem 3.1]{CB03} to nonparametric families. In this generality it seems that a quantitative bound like \eqref{eqn:universalityCover} is out of reach.

\begin{lemma} \label{lem:universality1}
Suppose the market satisfies Assumption \ref{ass:bound}. Let $\Theta$ be a totally bounded subset of $L^{\infty}(\Delta_n, \overline{\Delta}_n)$ and let $\nu_0$ be any initial distribution on $\Theta$ with full support. Then Cover's portfolio $\widehat{\pi}$ satisfies the universality property \eqref{eqn:universality}.
\end{lemma}
\begin{proof}
Since $\widehat{V}(t) \leq V^*(t)$ for all $t$, it suffices to show that $\liminf_{t \rightarrow \infty} \frac{1}{t} \log \frac{\widehat{V}(t)}{V^*(t)} \geq 0$. Let $\epsilon > 0$ be given. Then there exists $\epsilon' > 0$ and portfolios $\pi_1, \ldots \pi_N \in \Theta$ such that the set $\{\pi_j\}_{1 \leq j \leq N}$ are $\epsilon'$-dense in $\Theta$, and whenever $\|\pi - \pi_j\|_{\infty} < \epsilon'$ we have $|\ell_{\pi} - \ell_{\pi_j}| < \epsilon$ on the set ${\mathcal{S}}$ defined by \eqref{eqn:pairstatespace}.

For every $t > 0$, there exists a portfolio $\pi^{[t]} \in \Theta$ such that
\[
\frac{1}{t} \log V_{\pi^{[t]}}(t) > \frac{1}{t} \log V^*(t) - \epsilon,
\]
and from the above construction there exists $1 \leq j^{[t]} \leq N$ such that $\pi^{[t]} \in B_{j^{[t]}} := B(\pi_{j^{[t]}}, \epsilon')$, the open ball in $\Theta$ with radius $\epsilon'$ centered at $\pi_{j^{[t]}}$. Thus
\begin{equation} \label{eqn:approxbest}
\left|\frac{1}{t} \log V_{\pi_{j^{[t]}}}(t) - \frac{1}{t} \log V^*(t) \right| < 2\epsilon.
\end{equation}
Moreover, for all $\pi \in B_{j^{[t]}}$ we have
\[
\left|\frac{1}{t} \log V_{\pi}(t) - \frac{1}{t} \log V_{\pi_{j^{[t]}}}(t) \right| < \epsilon
\]
for all $t$. Exponentiating and combining these inequalities and using the triangle inequality, we have
\begin{equation}  \label{eqn:approxhat}
\begin{split}
\frac{1}{t} \log \widehat{V}(t) &\geq \frac{1}{t} \log \int_{  B_{j^{[t]}} } V_{\pi}(t) \nu_0(d\pi) \\
&= \frac{1}{t} \log \int_{B_{j^{[t]}}} \exp \left(t \cdot \frac{1}{t} \log V_{\pi}(t) \right) \nu_0(d\pi) \\
&\geq \frac{1}{t} \log \int_{ B_{j^{[t]}} } \exp \left(t \cdot \left( \frac{1}{t} \log V^*(t) - 3 \epsilon \right)\right) \nu_0(d\pi) \\
&\geq \frac{1}{t} \log V^*(t) - 3\epsilon + \frac{1}{t} \log \nu_0( B_{j^{[t]}} ) .
\end{split}
\end{equation}
Note that $j^{[t]}$ can only take finitely many values. Since $\nu_0$ has full support, we have
\[
\lim_{t \rightarrow \infty} \frac{1}{t} \log \nu_0( B_{j^{[t]}} ) = 0.
\]
It follows from \eqref{eqn:approxhat} that
\[
\liminf_{t \rightarrow \infty} \frac{1}{t} \log \frac{\widehat{V}(t)}{V^*(t)} \geq -3\epsilon.
\]
The proof is completed by letting $\epsilon \rightarrow 0$.
\end{proof}

Theorem \ref{thm:main1} is a consequence of Lemma \ref{lem:universality1} and the following `uniform strong law of large numbers'. The proof is a standard bracketing argument similar to the proof of Lemma \ref{lem:universality1} and can be found, for example, in \cite[Section 3.1]{G00}.

\begin{lemma} \label{lem:ULLN1}
Under the hypotheses of Theorem \ref{thm:main1}, we have
\[
\lim_{t \rightarrow \infty} \sup_{\pi \in \Theta} \left| \frac{1}{t} \log V_{\pi}(t) - W(\pi) \right| = 0.
\]
\end{lemma}

\medskip

\begin{proof}[Proof of Theorem \ref{thm:main1}]
By assumption, $W(\pi) = \lim_{t \rightarrow \infty} \frac{1}{t} \log V_{\pi}(t)$ exists for all $\pi \in \Theta$. Using the argument of the proof of Lemma \ref{lem:universality1}, it can be shown that
\begin{equation} \label{eqn:commonlimit}
\lim_{t \rightarrow \infty} \frac{1}{t} \log \widehat{V}(t) = \lim_{t \rightarrow \infty} \frac{1}{t} \log V^*(t) = W^*.
\end{equation}

Since
\begin{equation*}
\begin{split}
\frac{1}{t} \log \nu_t(B) &= \frac{1}{t} \log \left(\frac{1}{\widehat{V}(t)} \int_{\Theta} V_{\pi}(t) d\nu_0(\pi)\right) \\
  &= \frac{1}{t} \log \left(  \int_{\Theta} V_{\pi}(t) d\nu_0(\pi)\right) - \frac{1}{t} \log \widehat{V}(t)
\end{split}
\end{equation*}
and thanks to \eqref{eqn:commonlimit}, to prove the LDP it suffices to show that
\begin{equation} \label{eqn:upperbound1}
\limsup_{t \rightarrow \infty} \frac{1}{t} \log \int_F V_{\pi}(t) d\nu_0(\pi) \leq \sup_{\pi \in F} W(\pi)
\end{equation}
for all closed sets $F$ and
\begin{equation} \label{eqn:lowerbound1}
\liminf_{t \rightarrow \infty} \frac{1}{t} \log \int_G V_{\pi}(t) d\nu_0(\pi) \geq \inf_{\pi \in G} W(\pi)
\end{equation}
for all open sets $G$. Indeed, we will show that \eqref{eqn:upperbound1} holds for all measurable sets no matter it is closed or not.

By Lemma \ref{lem:ULLN1}, the quantity $R(t) = \sup_{\pi \in \Theta} \left| \frac{1}{t} \log V_{\pi}(t) - W(\pi) \right|$ converges to $0$ as $t \rightarrow \infty$. To prove the upper bound, write
\begin{equation*}
\begin{split}
\frac{1}{t} \log \int_F V_{\pi}(t) d\nu_0(\pi) &\leq \frac{1}{t} \log \int_F \exp\left(t\left(W(\pi) + R(t)\right)\right) d\nu_0(\pi) \\
   &\leq \sup_{\pi \in F} W(\pi) + R(t) + \frac{1}{t} \nu_0(F).
\end{split}
\end{equation*}
Letting $t \rightarrow \infty$ establishes the upper bound for all measurable sets. The lower bound for open sets can be proved in a similar manner using the fact that $\nu_0$ has full support.
\end{proof}

\begin{proof} [Proof of Theorem \ref{thm:LDPfinite}]
Since $E$ is finite, $\Theta$ is a totally bounded family of functions from $E$ to $\overline{\Delta}_n$. (It can be extended from $E$ to $\Delta_n$ by setting $\pi(p) = \pi_0$ for $p \notin E$, where $\pi_0$ is a fixed element of $\overline{\Delta}_n$.) The first statement then follows from Lemma \ref{lem:universality1}. Moreover, by Lemma \ref{lem:logoptimal} the limit $W(\pi) = \lim_{t \rightarrow \infty} \int_{E \times E} \ell_{\pi} d \mathbb{P}_t$ exists and equals $\int_{E \times E} \ell_{\pi} d \mathbb{P}$ for all $\pi \in \Theta$. Thus Theorem \ref{thm:main1} applies. It is easy to see that $I(\pi)$ is convex in $\pi$.
\end{proof}

\subsection{An example} \label{sec:example}
Theorem \ref{thm:main1} assumes that the family is totally bounded in the supremum metric and the asymptotic growth rates of all portfolios exist. Now we give a simple example to show what might go wrong. First, if the family is too large and the topology is not chosen appropriately, universality may fail. Second, the LDP may be trivial even if there is an optimal portfolio.

Consider a market with two stocks (so $n = 2$). Assume that the market weight takes values in the countable set
\[
E = \{p = (p_1, p_2) \in \Delta_2: p_1, p_2 \text{ rational}\}.
\]
Let $\Theta = \left(\overline{\Delta}_2\right)^E$ be the set of portfolio maps on $E$ and equip $\Theta$ with the topology of {\it pointwise convergence}. Let the initial distribution $\nu_0$ be the infinite product of the uniform distribution on $\overline{\Delta}_2$. That is, if $\pi$ is chosen randomly from $\Theta$ according to the distribution $\nu_0$, then for any $p^{(1)}, \ldots, p^{(k)} \in E$ the portfolio vectors $\pi(p^{(j)})$ are i.i.d.~uniform in $\Delta_2$. It is easy to verify that $\nu_0$ has full support on $\Theta$.

Let $\delta > 0$ be a rational number and consider the path $\{\mu(t)\}_{t \geq 0}$ in $E$ defined recursively by
\begin{equation} \label{eqn:examplemarket}
\mu(0) = \left(\frac{1}{2}, \frac{1}{2}\right), \quad \mu(t + 1) = \left(\frac{\mu_1(t)}{1 + \delta \mu_2(t)}, \frac{(1 + \delta) \mu_2(t)}{1 + \delta \mu_2(t)}\right).
\end{equation}
Note that
\begin{equation} \label{eqn:secondstockbetter}
\frac{\mu_2(t + 1)}{\mu_2(t)} = (1 + \delta) \frac{\mu_1(t + 1)}{\mu_1(t)}
\end{equation}
for all $t \geq 0$ and it can be verified directly that $\{\mu(t)\}_{t \geq 0}$ satisfies Assumption \ref{ass:bound} with $M = 1 + \delta$.

From \eqref{eqn:secondstockbetter}, it is clear that any optimal portfolio $\pi$ up to time $t$ satisfies $\pi(\mu(s)) = (0, 1)$ for all $0 \leq s \leq t - 1$. It follows that
\[
V^*(t) = \max_{\pi \in \Theta} V_{\pi}(t) = \frac{\mu_2(t)}{\mu_2(0)}, \quad t > 0.
\]

\begin{proposition} \label{lemma:nouniverality}
For the market weight path given by \eqref{eqn:examplemarket}, Cover's portfolio $\widehat{\pi}$ satisfies
\[
\widehat{V}(t) = \frac{\mu_2(t)}{\mu_2(0)} \left(1 - \frac{1}{2} \frac{\delta}{1 + \delta} \right)^t.
\]
In particular, we have
\[
\lim_{t \rightarrow \infty} \frac{1}{t} \log \frac{\widehat{V}(t)}{V^*(t)} = \log \left( 1 - \frac{1}{2} \frac{\delta}{1 + \delta}\right) < 0.
\]
Thus Cover's portfolio does not satisfy the universality property \eqref{eqn:universality} for all market weight paths satisfying Assumption \ref{ass:bound}.
\end{proposition}
\begin{proof}
Given a portfolio $\pi \in \Theta$, we have
\[
V_{\pi}(t) = \prod_{s = 0}^{t - 1} \left( \pi_1(\mu(s)) \frac{\mu_1(s + 1)}{\mu_1(s)} + \pi_2(\mu(s)) \frac{\mu_2(s + 1)}{\mu_2(s)}\right).
\]
By \eqref{eqn:secondstockbetter}, we can write
\begin{equation*}
\begin{split}
V_{\pi}(t) &= \prod_{s = 0}^{t - 1} \left( \frac{\mu_2(s + 1)}{\mu_2(s)}  \left( \frac{1}{1 + \delta} \pi_1(\mu(s)) + \pi_2(\mu(s)) \right)\right) \\
  &= \frac{\mu_2(t)}{\mu_2(0)} \prod_{s = 0}^{t - 1} \left(1 -   \left( 1- \pi_2(\mu(s))\right) \frac{\delta}{1 + \delta} \right).
\end{split}
\end{equation*}
The value of Cover's portfolio is
\begin{equation*}
\begin{split}
\widehat{V}(t) = \frac{\mu_2(t)}{\mu_2(0)} \int_{\Theta} \prod_{s = 0}^{t - 1}  \left(1 -   \left( 1- \pi_2(\mu(s))\right) \frac{\delta}{1 + \delta} \right) d\nu_0(\pi).
\end{split}
\end{equation*}
Since $\nu_0$ is the infinite product of uniform distributions, by independence we have
\[
\widehat{V}(t) = \frac{\mu_2(t)}{\mu_2(0)} \left(1 - \frac{1}{2} \frac{\delta}{1 + \delta}\right)^t. \qedhere
\]
\end{proof}

\begin{proposition}
For the market weight path given by \eqref{eqn:examplemarket}, the wealth distributions $\{\nu_t\}_{t = 0}^{\infty}$ satisfies LDP on $\Theta$ with the trivial rate function $I(\pi) \equiv 0$.
\end{proposition}
\begin{proof}
Let $G$ be any open set of $\Theta$. Then $G$ contains a cylinder set of the form
\begin{equation} \label{eqn:cylinder}
C = \left\{(\pi(p_1), \ldots, \pi(p_{\ell})) \in B\right\},
\end{equation}
where $p_1, \ldots, p_{\ell} \in E$ and $B$ is an open subset of $\left(\overline{\Delta}_2\right)^{\ell}$. It follows that
\[
\nu_t(G) \geq \frac{1}{ \left(1 - \frac{1}{2} \frac{\delta}{1 + \delta}\right)^t} \int_C \prod_{s = 0}^{t - 1}   \left(1 -   \left( 1- \pi_2(\mu(s))\right) \frac{\delta}{1 + \delta} \right) d \nu_0(\pi).
\]
Using the fact that $C$ puts restrictions on only finitely many coordinates, we have
\[
\lim_{t \rightarrow \infty} \frac{1}{t} \log \int_C \prod_{s = 0}^{t - 1}   \left(1 -   \left( 1- \pi_2(\mu(s))\right) \frac{\delta}{1 + \delta} \right) d \nu_0(\pi) = \log \left( 1 - \frac{1}{2} \frac{\delta}{1 + \delta}\right).
\]
Thus $\liminf_{t \rightarrow \infty} \frac{1}{t} \log \nu_t(G) \geq 0$ and $\lim_{t \rightarrow \infty} \frac{1}{t} \log \nu_t(G) = 0$. Since the upper bound holds trivially, the LDP is proved.
\end{proof}

\section{Functionally generated portfolios} \label{sec:fgp}
This section is devoted to proving Theorem \ref{thm:main2} for functionally generated portfolios. As in Section \ref{sec:discrete} we impose Assumption \ref{ass:bound} on the market weight sequence $\{\mu(t)\}_{t = 0}^{\infty}$. We begin by stating some properties of functionally generated portfolios introduced in Section \ref{sec:summary}. For convex analytic concepts a standard reference is \cite{R70}.

\subsection{Functionally generated portfolios}
First we give the convex analytic interpretation of the defining inequality \eqref{eqn:fgpineq}. Let $\Phi: \Delta_n \rightarrow {\Bbb R}$ be a concave function. The {\it superdifferential} $\partial \Phi(p)$ of $\Phi$ at $p \in \Delta_n$ is the convex set of all vectors $\xi \in {\Bbb R}^n$ satisfying $\sum_{i = 1}^n \xi_i = 0$ (i.e., $\xi$ is tangent to $\Delta_n$) and
\[
\Phi(p) + \langle \xi, q - p \rangle \geq \Phi(q)
\]
for all $q \in \Delta_n$. The elements of $\partial \Phi(p)$ are called {\it supergradients} of $\Phi$ at $p$. Note that if $\Phi$ is positive and concave, $\log \Phi$ is also a concave function.

\begin{lemma} \cite[Proposition 6]{PW14} \label{lem:FGPsuperdiff} {\ }
\begin{enumerate}
\item[(i)] Suppose $\pi$ is generated by $\Phi$. For every $p \in \Delta_n$, the tangent vector $v = (v_1, \ldots, v_n)$ of $\Delta_n$ given by
\[
v_i = \frac{\pi_i(p)}{p_i} - \frac{1}{n} \sum_{j = 1}^n \frac{\pi_j(p)}{p_j}
\]
is an element of $\partial \log \Phi(p)$, the superdifferential of $\log \Phi$ at $p$.
\item[(ii)] Conversely, suppose $\Phi$ is a positive concave function on $\Delta_n$. For each $p \in \Delta_n$, let $v(p)$ be an element of $\partial \log \Phi(p)$ and define $\pi(p)$ by
\[
\pi_i(p) = p_i \left(v_i(p) + 1 - \sum_{j = 1}^n p_j v_j(p) \right).
\]
Then $\pi$ is a map from $\Delta_n$ to $\overline{\Delta}_n$ and is a portfolio generated by $\Phi$. (By \cite[Theorem 14.56]{RW98}, there exists a measurable selection of $\partial \log \Phi$.)
\end{enumerate}
\end{lemma}

If $\Phi$ is not differentiable at $p$, the superdifferential $\partial \log \Phi(p)$ is an infinite set, and by Lemma \ref{lem:FGPsuperdiff} there are multiple ways to choose a portfolio generated by $\Phi$. Nevertheless, it is well known that a finite-valued concave function on $\Delta_n$ is differentiable almost everywhere on $\Delta_n$, so the portfolio maps generated by $\Phi$ agree almost everywhere on $\Delta_n$. Note, however, that the null set depends on $\Phi$. In general, a functionally generated portfolio $\pi: \Delta_n \rightarrow \overline{\Delta}_n$ is not continuous on $\Delta_n$.

Let ${\mathcal{FG}} \subset L^{\infty}\left(\Delta_n, \overline{\Delta}_n\right)$ be the family of all functionally generated portfolios $\pi: \Delta_n \rightarrow \overline{\Delta}_n$. It is known that ${\mathcal{FG}}$ is convex. Indeed, if $\pi$ is generated by $\Phi$ and $\eta$ is generated by $\Psi$, then for any $\lambda \in (0, 1)$ the portfolio $\lambda \pi + (1 - \lambda) \eta$ (a constant-weighted portfolio of $\pi$ and $\eta$ is generated by the geometric mean $\Phi^{\lambda} \Psi^{1 - \lambda}$. We endow ${\mathcal{FG}}$ with the topology of uniform convergence. The following lemma shows that the current setting is not covered by Theorem \ref{thm:main1}.

\begin{lemma}
${\mathcal{FG}}$ is not totally bounded. In fact, ${\mathcal{FG}}$ is not separable.
\end{lemma}
\begin{proof}
We give an example for $n = 2$ and similar considerations can be applied to all dimensions. For each $\theta \in (0, 1)$, let $\pi_{\theta}: \Delta_2 \rightarrow \overline{\Delta}_2$ be the portfolio
\[
\pi_{\theta}(p) =
\begin{cases} (1, 0) &\mbox{if } p_1 \leq \theta \\
(0, 1) & \mbox{if } p_1 > \theta. \end{cases}
\]
It is easy to verify that $\pi_{\theta}$ is functionally generated, and the generating function is the smallest piecewise affine function $\Phi_{\theta}$ on $\Delta_2$ satisfying $\Phi_{\theta}((0, 1)) = \Phi_{\theta}((1, 0)) = 0$ and $\Phi_{\theta}(\theta, 1 - \theta) = 1$. Since $\left\{\pi_{\theta}\right\}_{\theta \in (0, 1)}$ forms an uncountable discrete set in ${\mathcal{FG}}$, ${\mathcal{FG}}$ is not separable.
\end{proof}

Although the portfolio maps  $\pi: \Delta_n \rightarrow \overline{\Delta}_n$ are the primary objects, it is technically more convenient to work with their generating functions.

\begin{definition} \label{def:metric}
Let ${\mathcal{C}}_0$ be the set of all positive concave functions $\Phi$ on $\Delta_n$ satisfying the normalization $\Phi\left(\overline{e}\right) = 1$, where $\overline{e} = \left(\frac{1}{n}, \ldots, \frac{1}{n}\right)$ is the barycenter of $\overline{\Delta}_n$. We endow ${\mathcal{C}}_0$ with the topology of local uniform convergence. We define a metric $d$ on ${\mathcal{C}}_0$ as follows. For $m  = 1, 2, .\ldots$, let $K_m = \left\{p \in \Delta_n: p_i \geq \frac{1}{m}, 1 \leq i \leq n\right\}$. Then $\{K_m\}_{m = 1}^{\infty}$ is a compact exhaustion of $\Delta_n$. For $\Phi, \Psi \in {\mathcal{C}}_0$ we define
\[
d(\Phi, \Psi) = \sum_{m = 1}^{\infty} 2^{-m} \frac{\max_{p \in K_m} |\Phi(p) - \Psi(p)|}{1 + \max_{p \in K_m} |\Phi(p) - \Psi(p)|}.
\]
\end{definition}

By \cite[Proposition 6]{PW14} the generating function of a functionally generated portfolio is unique up to a positive multiplicative constant. Thus by a normalization we may assume without loss of generality that ${\mathcal{C}}_0$ is the set of generating functions.

\begin{lemma} \label{lem:compact}
$\left({\mathcal{C}}_0, d\right)$ is a compact metric space.
\end{lemma}
\begin{proof}
See \cite[Lemma 10]{W14}.
\end{proof}

Although ${\mathcal{FG}}$ is not totally bounded, by Lemma \ref{lem:FGPsuperdiff} and Lemma \ref{lem:compact} it is `almost the same' as ${\mathcal{C}}_0$ which is a compact metric space. This allows us to show under appropriate conditions that $V_{\pi}(t)$ behaves nicely as a function of $\pi$ when $t$ is large. Here is an application of the compactness of ${\mathcal{C}}_0$.

\begin{lemma}
For each $t \geq 0$, there exists $\pi^* \in \Theta$ such that $V_{\pi^*}(t) = \sup_{\pi \in {\mathcal{FG}}} V_{\pi}(t)$.
\end{lemma}
\begin{proof}
The proof is essentially the one in \cite[Theorem 4(i)]{W14} and is included for completeness. Let $\{\pi_k\}_{k = 1}^{\infty}$ be a maximizing sequence, i.e.,
\[
\sup_{\pi \in {\mathcal{FG}}} V_{\pi}(t) = \lim_{k \rightarrow \infty} V_{\pi_k}(t) =  \lim_{k \rightarrow \infty} \prod_{s = 0}^{t - 1} \left(\pi_k(\mu(s)) \cdot \frac{\mu(s + 1)}{\mu(s)}\right).
\]
Let $\{\Phi_k\}_{k = 1}^{\infty} \subset {\mathcal{C}}_0$ be the corresponding generating functions. By the compactness of ${\mathcal{C}}_0$ we may pass to a subsequence so that $\Phi_k \rightarrow \Phi \in {\mathcal{C}}_0$ locally uniformly on $\Delta_n$. We may pass to a further subsequence such that the limit $\lim_{k \rightarrow \infty} \pi_k(\mu(s))$ exists in $\overline{\Delta}_n$ for all $0 \leq s \leq t - 1$.

Let $\pi^*$ be a portfolio generated by $\Phi$ which exists by Lemma \ref{lem:FGPsuperdiff}. We claim that if we redefine $\pi^*$ on $\{\mu(s): 0 \leq s \leq t - 1\}$ by setting
\[
\pi^*(\mu(s)) = \lim_{k \rightarrow \infty} \pi_k(\mu(s))
\]
for $0 \leq s \leq t - 1$, then $\pi^*$ is still generated by $\Phi$ and so is an element of $\Theta$. By \eqref{eqn:fgpineq} it suffices to check that
\begin{equation} \label{eqn:checksuperdiff}
\pi^*(\mu(s)) \cdot \frac{q}{\mu(s)} \geq \frac{\Phi(q)}{\Phi(\mu(s))}
\end{equation}
for all $0 \leq s \leq t - 1$ and $q \in \Delta_n$. Now since $\pi_k$ is generated by $\Phi_k$, we have
\[
\pi_k(\mu(s)) \cdot \frac{q}{\mu(s)} \geq \frac{\Phi_k(q)}{\Phi_k(\mu(s))}.
\]
Letting $k \rightarrow \infty$, we get \eqref{eqn:checksuperdiff} and so $\pi^*$ is generated by $\Phi$. The lemma follows by noting that $V_{\pi^*}(t) = \lim_{k \rightarrow \infty} V_{\pi_k}(t)$.
\end{proof}

Continuing the statistical analogy (see Remark \ref{rem:stat}), the portfolio $\pi^*$ may be viewed as the maximum likelihood estimator of the portfolio which maximizes the asymptotic growth rate $W(\pi) = \lim_{t \rightarrow \infty} \frac{1}{t} \log V_{\pi}(t)$.

\begin{lemma} \label{lem:superdiffuniform}
Let $\Phi_0 \in {\mathcal{C}}_0$ and $p_0 \in \Delta_n$. Let $K \subset \Delta_n$ be a compact set whose (relative) interior contains $p_0$. Then for any $\epsilon > 0$, there exists $\delta > 0$ such that whenever $\Phi \in {\mathcal{C}}_0$, $\max_{p \in K} \left| \Phi(p) - \Phi_0(p) \right| < \delta$ and $|q - p_0| < \delta$, we have
\[
\partial \log \Phi(q) \subset \partial \log  \Phi(p_0) + \epsilon \overline{B}(0, 1).
\]
\end{lemma}
\begin{proof}
This is a uniform version of \cite[Theorem 24.5]{R70}. We will proceed by contradiction. If the statement is false, there exists $\epsilon_0 > 0$ such that the following holds. For every $k \geq 1$, there exists $\Phi_k \in {\mathcal{C}}_0$ and $p_k \in \Delta_n$ such that
\[
\max_{p \in K} \left| \Phi_k(p) - \Phi_0(p) \right| < \frac{1}{k}, \quad \left|p_k - p_0\right| < \frac{1}{k}
\]
and
\[
\partial\log   \Phi_k(p_k) \not\subset \partial \log  \Phi(p_0) + \epsilon_0 \overline{B}(0, 1).
\]
This contradicts \cite[Theorem 24.5]{R70} and thus the lemma is proved.
\end{proof}

Using Lemma \ref{lem:superdiffuniform} and Proposition \ref{lem:FGPsuperdiff} we have the following corollary which is a refined version of \cite[Lemma 11]{W14}.

\begin{lemma} \label{eqn:diffcontinuity}
Let $\pi_0$ be a portfolio generated by $\Phi_0$. Let $p_0 \in \Delta_n$ be a point at which $\Phi_0$ is differentiable. For any $\epsilon > 0$ and any compact neighborhood $K$ of $p_0$ in $\Delta_n$, there exists $\delta > 0$ such that whenever $\pi$ is generated by $\Phi$ and $\max_{p \in K} \left| \Phi(p) - \Phi_0(p) \right| < \delta$, we have $\max_{p: |p - p_0| < \delta} |\pi(p) - \pi_0(p_0)| < \epsilon$.
\end{lemma}

We end this subsection with some technical remarks.

\begin{remark}
It is natural to ask why we do not use the compact set ${\mathcal{C}}_0$ as the index set. There are three reasons for this. First, the portfolio maps $\pi: \Delta_n \rightarrow \overline{\Delta}_n$ are the primary objects for portfolio analysis, and the generating functions are only derived entities. Second, even if $\pi_1$ and $\pi_2$ have the same generating function $\Phi$, over a finite horizon $V_{\pi_1}(t)$ and $V_{\pi_2}(t)$ may have quite different behaviors. This is because it may happen that the market lands repeatedly at the points where $\Phi$ is not differentiable and the two portfolios differ. Third, even though for each $\Phi \in {\mathcal{C}}_0$ we may choose a portfolio $\pi_{\Phi}$ generated by $\Phi$, there is no canonical way of doing this so that the maps $\Phi \mapsto \pi_{\Phi}$ and $\Phi \mapsto V_{\pi_{\Phi}}(t)$ are measurable.
\end{remark}

\subsection{Asymptotic growth rate}
Recall from Lemma \ref{lem:integral} that $V_{\pi}(t)$ can be written in the form $\frac{1}{t} \log V_{\pi}(t) = \int_S \ell_{\pi} d {\Bbb P}_t$, where $\ell_{\pi}(p, q) = \log \left( \pi(p) \cdot \frac{q}{p} \right)$ is defined in \eqref{eqn:returnfunction} and
\[
{\Bbb P}_t = \frac{1}{t} \sum_{s = 0}^{t - 1} \delta_{(\mu(s), \mu(s + 1))}
\]
is the empirical measure of the pair $(\mu(s), \mu(s + 1))$ up to time $t$. Appealing to the long term stability of capital distribution, we assume that ${\Bbb P}$ converges weakly to an absolutely continuous probability measure ${\Bbb P}$. We denote by $B(p, \delta)$ the Euclidean ball in $\Delta_n$ centered at $p$ with radius $\delta$. The Euclidean norm is denoted by $|\cdot|$.

First we prove a `strong law of large numbers' for individual elements of ${\mathcal{FG}}$. We will use some basic results of the theory of weak convergence \cite{B09}. Recall that a ${\Bbb P}$-continuity set is a set $A$ satisfying ${\Bbb P}\left(\partial A \right) = 0$, where $\partial A$ is the boundary of $A$. We write $\partial_{{\mathcal{S}}} A$ if we want to be explicit about the underlying topological space.

\begin{lemma} \label{lem:LLN}
Suppose ${\Bbb P}_t$ converges weakly to an absolutely continuous probability measure ${\Bbb P}$ on ${\mathcal{S}}$. Then for every $\pi \in {\mathcal{FG}}$ the asymptotic growth rate $W(\pi) = \lim_{t \rightarrow \infty} \frac{1}{t} \log V_{\pi}(t)$ exists and is given by
\begin{equation} \label{eqn:weakconvergence}
W(\pi) = \lim_{t \rightarrow \infty} \int_{\mathcal{S}} \ell_{\pi} d {\Bbb P}_t = \int_{\mathcal{S}} \ell_{\pi} d {\Bbb P}.
\end{equation}
\end{lemma}
\begin{proof}
Note that \eqref{eqn:weakconvergence} does not follow directly from the definition of weak convergence because $\ell_{\pi}$ may have discontinuities. The constructions here (refined from the proof of \cite[Theorem 5]{W14}) will be useful when we prove uniform convergence in Lemme \ref{lem:ULLN}.

Let $\epsilon > 0$ be given. Let $\Phi \in {\mathcal{C}}_0$ be the generating function of $\pi$ and consider the set
\[
D = \{p \in \Delta_n: \Phi \text{ is differentiable at } p\}.
\]
Then $\Delta_n \setminus D$ has Lebesgue measure $0$. Given $\epsilon$, there exists $\epsilon' > 0$ such that whenever $\pi_1, \pi_2 \in \overline{\Delta}_n$ and $|\pi_1 - \pi_2| < \epsilon'$, we have
\begin{equation} \label{eqn:returnfunctionbound}
\left| \log \left( \pi_1 \cdot \frac{q}{p} \right) - \log \left( \pi_2 \cdot \frac{q}{p} \right) \right| < \epsilon
\end{equation}
for all $(p, q) \in {\mathcal{S}}$.

For each $p \in D$, by Lemma \ref{eqn:diffcontinuity} there exists $\delta(p) > 0$ such that $B(p, \delta(p)) \subset \Delta_n$ and $|q - p| < \delta(p)$ implies
\begin{equation} \label{eqn:weightdifference}
|\pi(q) - \pi(p)| < \epsilon'.
\end{equation}

As a subspace of a separable metric space, $D$ is separable. Hence, there exists a countable set $\{p_k\}_{k = 1}^{\infty} \subset D$ such that
\[
D \subset \bigcup_{k = 1}^{\infty} B(p_k, \delta(p_k)).
\]
Let $A_1 = B(p_1, \delta(p_1))$ and for $k \geq 2$ define
\[
A_k = B(p_k, \delta(p_k)) \setminus \bigcup_{j = 1}^{k-1} B(p_j, \delta(p_j)).
\]
Then the sets $\{A_k\}$ are disjoint and
\[
\left(D \times \Delta_n \right) \cap {\mathcal{S}} \subset \bigcup_{k = 1}^{\infty} \left(A_k \times \Delta_n\right) \cap {\mathcal{S}}.
\]

Since ${\Bbb P} \left( \left(D \times \Delta_n \right) \cap {\mathcal{S}} \right) = 1$ by absolute continuity, by continuity of measure there exists a positive integer $k_0$ such that
\[
{\Bbb P} \left( \bigcup_{k = 1}^{k_0} \left(A_k \times \Delta_n\right) \cap {\mathcal{S}}\right) > 1 - \epsilon.
\]
Define
\[
A_0 = \Delta_n \setminus \left( \bigcup_{k = 1}^{k_0} A_k \right).
\]
Then
\begin{equation} \label{eqn:remaindermeasure}
{\Bbb P}\left((A_0 \times \Delta_n) \cap {\mathcal{S}}\right) \leq \epsilon.
\end{equation}

Note that for $0 \leq k \leq k_0$, $\left(A_k \times \Delta_n\right) \cap {\mathcal{S}}$ is a ${\Bbb P}$-continuity set as it is formed by set-theoretic operations on ${\mathcal{S}}$ (which has piecewise smooth boundary), $\Delta_n$ and Euclidean balls. Also, by Assumption \ref{ass:bound} $|\ell_{\pi}(\cdot, \cdot)|$ is bounded uniformly on ${\mathcal{S}}$ by $M' := \log M$. So, for each $1 \leq k \leq k_0$ the map
\[
(p, q) \mapsto \ell_{\pi(p(k))}\left(p, q\right) := \log \left(\pi(p(k)) \cdot \frac{q}{p}\right)
\]
is a bounded continuous function on ${\mathcal{S}}$.

By weak convergence and Lemma \ref{lem:weakconvergence} in the Appendix, there exists a positive integer $t_0$ such that for $t \geq t_0$ we have
\begin{equation} \label{eqn:inequalityepsilon}
{\Bbb P}_t \left( \left( A_0 \times \Delta_n \right) \cap {\mathcal{S}} \right)  < 2\epsilon
\end{equation}
and
\begin{equation} \label{eqn:lemconsequence}
\left|\int_{(A_k \times \Delta_n) \cap {\mathcal{S}}} \ell_{\pi(p(k))} d ({\Bbb P}_t - {\Bbb P}) \right| < \frac{\epsilon}{k_0}.
\end{equation}
(note that $k_0$ is fixed before $t_0$ is chosen).

Now we estimate the difference $\left|\frac{1}{t} \log V_{\pi}(t) - \int_{{\mathcal{S}}} \ell_{\pi} d {\Bbb P}\right| = \left| \int_{{\mathcal{S}}} \ell_{\pi} d  \left( {\Bbb P}_t - {\Bbb P} \right) \right|$. We have
\begin{equation} \label{eqn:mainestimate1}
\begin{split}
\left| \int_{{\mathcal{S}}} \ell_{\pi} d  \left( {\Bbb P}_t - {\Bbb P} \right) \right|
&\leq  \left|\sum_{k = 1}^{k_0} \int_{(A_k \times \Delta_n) \cap {\mathcal{S}}} \ell_{\pi} d\left( {\Bbb P}_t - {\Bbb P}\right)\right| +
\left| \int_{(A_0 \times \Delta_n) \cap{\mathcal{S}}} \ell_{\pi} d\left( {\Bbb P}_t - {\Bbb P}\right) \right|.
\end{split}
\end{equation}

Using the boundedness of $\ell_{\pi}$,  \eqref{eqn:remaindermeasure} and \eqref{eqn:inequalityepsilon}, the second term of \eqref{eqn:mainestimate1} is bounded by $3M' \epsilon$. Now for each $k$, by \eqref{eqn:returnfunctionbound}, \eqref{eqn:weightdifference} and \eqref{eqn:lemconsequence} we have
\begin{equation*}
\begin{split}
\left| \int_{(A_k \times \Delta_n) \cap {\mathcal{S}}} \ell_{\pi} d\left( {\Bbb P}_t - {\Bbb P}\right)\right| &\leq \int_{(A_k \times \Delta_n) \cap {\mathcal{S}}}  \left| \ell_{\pi} - \ell_{\pi(p_k)} \right| d {\Bbb P}_t \\
& \ \ \ + \int_{(A_k \times \Delta_n) \cap {\mathcal{S}}}  \left| \ell_{\pi} - \ell_{\pi(p_k)} \right| d {\Bbb P} \\
& \ \ \ + \left|  \int_{(A_k \times \Delta_n) \cap {\mathcal{S}}} \ell_{\pi(p_k)} d \left( {\Bbb P}_t - {\Bbb P} \right) \right| \\
&\leq \epsilon {\Bbb P}_t \left((A_k \times \Delta_n) \cap {\mathcal{S}}\right) + \epsilon  {\Bbb P} \left((A_k \times \Delta_n) \cap {\mathcal{S}}\right) + \frac{\epsilon}{k_0}.
\end{split}
\end{equation*}
Summing the above inequality over $k$, we get
\[
\left| \int_S \ell_{\pi} d \left({\Bbb P}_t - {\Bbb P}\right) \right| \leq \epsilon + \epsilon + \epsilon + 3M'\epsilon, \quad t \geq t_0,
\]
and the lemma is proved.
\end{proof}

\subsection{Glivenko-Cantelli property}
Now we observe that the proof of Lemma \ref{lem:LLN} can be modified to yield a uniform version which implies Theorem \ref{thm:main2}(i). Recall that $d(\Phi, \Psi)$ is the metric on ${\mathcal{C}}_0$ given in Definition \ref{def:metric}.

\begin{lemma} \label{lem:ULLN}
Suppose ${\Bbb P}_t$ converges weakly to an absolutely continuous probability measure ${\Bbb P}$ on ${\mathcal{S}}$. Let $\pi_0 \in {\mathcal{FG}}$ be generated by $\Phi_0 \in {\mathcal{C}}_0$. For any $\epsilon > 0$, there exists $\delta > 0$ such that
\begin{equation} \label{eqn:localuniformFG}
\limsup_{t \rightarrow \infty} \sup_{\pi \in {\mathcal{FG}}(\pi_0, \delta)} \left| \frac{1}{t} \log V_{\pi}(t)  - \frac{1}{t} \log V_{\pi_0}(t) \right| < \epsilon,
\end{equation}
where ${\mathcal{FG}}(\pi_0, \delta)$ is the set of all functionally generated portfolio $\pi$ whose generating function $\Phi \in {\mathcal{C}}_0$ satisfies $d(\Phi, \Phi_0) < \delta$. In particular, we have the `uniform strong law of large numbers'
\begin{equation} \label{eqn:ULLN}
\lim_{t \rightarrow \infty} \sup_{\pi \in {\mathcal{FG}}} \left| \frac{1}{t} \log V_{\pi}(t) - W(\pi) \right| = 0.
\end{equation}
\end{lemma}
\begin{proof}
We want to estimate
\[
 \sup_{\pi \in {\mathcal{FG}}(\pi_0, \delta)} \left|\frac{1}{t} \log V_{\pi}(t) - \frac{1}{t} \log V_{\pi_0}(t)\right| =  \sup_{\pi \in {\mathcal{FG}}(\pi_0, \delta)} \left| \int_S \left(\ell_{\pi} - \ell_{\pi_0} \right) d{\Bbb P}_t\right|.
\]
Recall from Definition \ref{def:metric} that $K_m = \left\{p \in \Delta_n: p_i \geq \frac{1}{m}\right\}$. By continuity of measure, we can choose $m$ so that
\[
{\Bbb P} \left((K_m \times \Delta_n) \cap S\right) > 1 - \epsilon.
\]
Since $(K_m \times \Delta_n) \cap {\mathcal{S}}$ is a ${\Bbb P}$-continuity set, for $t$ sufficiently large we have
\[
\left|\left( \int_{{\mathcal{S}}} - \int_{(K_m \times \Delta_n) \cap {\mathcal{S}}} \right)\left(\ell_{\pi} - \ell_{\pi_0} \right) d{\Bbb P}_t \right| < 4M \epsilon,
\]
where $M' = \log M$ is the upper bound of $|\ell_{\pi}|$ and $|\ell_{\pi_0}|$ on ${\mathcal{S}}$. This allows us to focus on the set $\left(K_m \cap \Delta_n\right) \cap {\mathcal{S}}$.

Fix $\epsilon' > 0$. By Lemma \ref{eqn:diffcontinuity}, for each $p$ in the (relative) interior of $K_m$ at which $\Phi_0$ is differentiable (call this set $D_m$), there exists $\delta'(p) > 0$ such that whenever $\max_{q \in K_m} |\Phi(q) - \Phi_0(q)| < \delta'(p)$ and $|q - p| < \delta'(p)$, we have $|\pi(q) - \pi_0(p)| < \epsilon'$.

As in the proof of Lemma \ref{lem:LLN}, we may cover $D_m$ by a disjoint countable union $\bigcup_{k = 1}^{\infty} A_k$, where $A_k$ is a ${\Bbb P}$-continuity set containing $p_k$ and has diameter bounded by $\delta'(p_k)$.

Now choose a positive integer $k_0$ such that
\[
{\Bbb P} \left( \left( \bigcup_{k = 1}^{k_0} A_k \times \Delta_n \right) \cap {\mathcal{S}} \right) > 1 - 2 \epsilon.
\]
Also, choose $\delta > 0$ such that
\[
d(\Phi, \Phi_0) < \delta \Rightarrow \max_{p \in K_m} |\Phi(p) - \Phi_0(p)| < \min_{1 \leq k \leq k_0} \delta'(p_k).
\]
It follows that
\[
\sup_{\pi \in {\mathcal{FG}}(\pi_0, \delta(p_k))} \sup_{p: |p - p_k| < \delta'(p_k)}  \left|\pi(p) - \pi_0(p_k)\right| < \epsilon',
\]
With this uniform local approximation, we may follow the same steps as the proof of Lemma \ref{lem:LLN} to prove that
\[
\limsup_{t \rightarrow \infty} \sup_{\pi \in {\mathcal{FG}}(\pi_0, \delta)} \left| \frac{1}{t} \log V_{\pi}(t)  - W(\pi) \right| < C\epsilon,
\]
where $C > 0$ is a constant. Thus \eqref{eqn:localuniformFG} follows by letting $\epsilon \rightarrow 0$.

Note that \eqref{eqn:localuniformFG} implies that $\sup_{\pi \in {\mathcal{FG}}(\pi_0, \delta)} \left|W(\pi) - W(\pi_0) \right| \leq \epsilon$. Since ${\mathcal{C}}_0$ is compact, we may cover ${\mathcal{FG}}$ by finitely many sets of the form ${\mathcal{FG}}(\pi_0, \delta)$, and \eqref{eqn:ULLN} follows.
\end{proof}

\subsection{LDP and universality}
Now we finish the proof of Theorem \ref{thm:main2}. Recall that $\widehat{V}(t) = \int_{\Theta} V_{\pi}(t) d\nu_0(\pi)$ and $V^*(t) = \sup_{\pi \in \Theta} V_{\pi}(t)$.

\begin{lemma} \label{lem:Covergrowthrate}
Suppose ${\Bbb P}_t$ converges weakly to an absolutely continuous probability measure ${\Bbb P}$ on ${\mathcal{S}}$. Let $\nu_0$ be any initial distribution on ${\mathcal{FG}}$ and $W^* = \sup_{\pi \in \mathrm{supp}(\nu_0)} W(\pi)$. Then $\lim_{t \rightarrow \infty} \frac{1}{t} \log \widehat{V}(t) = W^*$.
\end{lemma}
\begin{proof}
For $\pi \in {\mathcal{FG}}$ write
\[
\frac{1}{t} \log V_{\pi}(t) = W(\pi) + R_{\pi}(t)
\]
where $R_{\pi}(t)$ is the remainder. By Lemma \ref{lem:ULLN} we have $\lim_{t \rightarrow \infty} \sup_{\pi \in {\mathcal{FG}}} \left| R_{\pi}(t) \right| = 0$. Write
\[
\widehat{V}(t) = \int_{\mathrm{supp}(\nu_0)} e^{t(W(\pi) + R_{\pi}(t))} d \nu_0(\pi).
\]
It is clear that $\limsup_{t \rightarrow \infty} \frac{1}{t} \log \widehat{V}(t) \leq W^*$. To show the other inequality, note that $W(\pi)$ is continuous in $\pi \in {\mathcal{FG}}$. Thus for any $\pi \in \supp(\nu_0)$ and $\epsilon > 0$, by restricting the integral to a neighborhood of $\pi$ we have $\liminf_{t \rightarrow \infty} \frac{1}{t} \log \widehat{V}(t) \geq W(\pi) - \epsilon$. Taking supremum over $\pi \in \supp(\nu_0)$ completes the proof.
\end{proof}

\begin{proof}[Proof of Theorem \ref{thm:main2}]
(i) This has been proved in Lemma \ref{lem:ULLN}.

(ii) We argue as in the proof of Theorem \ref{thm:main2}. Write
\[
\nu_t(B) = \frac{1}{\widehat{V}(t)} \int_{B \cap \supp(\nu_0)} V_{\pi}(t) d\nu_0(\pi).
\]
Using the uniform convergence property (i), we can show that
\begin{equation} \label{eqn:upperbound2}
\limsup_{t \rightarrow \infty} \frac{1}{t} \log \int_F V_{\pi}(t) d\nu_0(\pi) \leq \sup_{\pi \in F \cap \supp(\nu_0)} W(\pi)
\end{equation}
for any set $F$ with $F \cap \supp(\nu_0) \neq \emptyset$, and
\begin{equation} \label{eqn:lowerbound2}
\liminf_{t \rightarrow \infty} \frac{1}{t} \log \int_G V_{\pi}(t) d\nu_0(\pi) \geq \inf_{\pi \in G \cap \supp(\nu_0)} W(\pi)
\end{equation}
for all open sets $G$ such that $G \cap \supp(\nu_0) \neq \emptyset$. These inequalities and Lemma \ref{lem:Covergrowthrate} imply the LDP.

(iii) Let $\{\Phi_k\}_{k = 1}^{\infty}$ be a countable dense set in the metric space $\left({\mathcal{C}}_0, d\right)$. For each $k$, let $\pi_k$ be a portfolio generated by $\Phi_k$. Consider an initial distribution of the form
\begin{equation} \label{eqn:mynu}
\nu_0 = \sum_{k = 1}^{\infty} \lambda_k \delta_{\pi_k},
\end{equation}
where $\lambda_k > 0$ and $\sum_{k = 1}^{\infty} \lambda_k = 1$.

To see that $\nu_0$ works, let $\pi$ be any functionally generated portfolio and $\Phi \in {\mathcal{C}}_0$ be its generating function. Then there is a sequence $\pi_{k'}$ whose generating functions $\Phi_{k'}$ converges locally uniformly to $\Phi$. By Lemma \ref{lem:ULLN}, we have $W(\pi_{k'}) \rightarrow W(\pi)$. Thus $W^* = \sup_{\pi \in \supp(\nu_0)} W(\pi) = \sup_{\pi \in {\mathcal{FG}}} W(\pi)$. By Lemma \ref{lem:Covergrowthrate}, to establish the asymptotic universality property \eqref{eqn:thmuniversality} it remains to show that
\[
\lim_{t \rightarrow \infty} \frac{1}{t} \log V^*(t) = W^*,
\]
but this is a direct consequence of the uniform convergence property (i).
\end{proof}

\section{Conclusion and further problems} \label{sec:conclusion}
In this paper we studied Cover's portfolio from the point of view of stochastic portfolio theory. Given a family of portfolios, we studied its wealth distribution which is analogous to the capital distribution of an equity market. In this setting, the wealth distribution is not stable and diverse in the sense of stochastic portfolio theory, and under certain conditions we quantified its concentration in terms of large deviation principles. We also extended Cover's portfolio to the nonparametric family of functionally generated portfolios and established its asymptotic universality in the spirit of \cite{J92}.

Similar to \cite{J92} and \cite{GLU06}, the results in this paper are asymptotic in nature, and in this nonparametric setting we are unable to establish quantitative bounds that hold for all finite horizons. It is desirable to obtain quantitative bounds despite of the fact that they may be too conservative to be useful in practice. Even if the underlying market process is modeled correctly, the convergence $\frac{1}{t} \log V_{\pi}(t) \rightarrow W(\pi)$ may take a long time and the portfolio $\widehat{\pi}(t)$ may be dominated by noise. A possible remedy is to use a smaller family or to impose regularization via a suitable prior (initial distribution). Tackling this bias-variance trade-off in dynamic portfolio selection is an interesting problem of great practical importance.

\begin{problem}
For Cover's portfolio for the family of functionally generated portfolios, is it possible to choose an initial distribution such that $\widehat{\pi}$ can be computed or approximated numerically and a quantitative lower bound of $\widehat{V}(t) / V^*(t)$ can be proved?
\end{problem}

A possible direction is to restrict to functionally generated portfolios that are {\it rank-based}, that is, the portfolio weight of a stock depends only on its rank according to capitalization. Equivalently, this means that the generating functions are invariant under relabeling of coordinates. This has the effect of reducing the effective domain of $\pi$ and $\Phi$ to $\frac{1}{n!}$ of the unit simplex $\Delta_n$. By reducing the curse of dimensionality, we may be able to obtain an better bound.

Instead of using Cover's portfolio as a wealth-weighted average, we may use other portfolio selection algorithms to construct universal portfolios for functionally generated portfolios. Perhaps the {\it follow-the-regularized-leader} (FTRL) approach of \cite{HK15} can be generalized to this nonparametric set up. Intuitively, we want to perform a sort of online gradient descent on the set ${\mathcal{FG}}$.

A classic result in asymptotic parametric statistics is the {\it Bernstein von-Mises Theorem} which states that the posterior distribution is asymptotically normal under appropriate scaling \cite[Chapter 10]{V00}. Certain generalizations to nonparametric models are possible, see for example \cite{CN13}. As noted in the Introduction, for constant-weighted portfolios the map $\pi \mapsto V_{\pi}(t)$ is essentially a multiple of a normal density (see \cite{J92} and \cite{CB03}). Hence the wealth distribution, when suitably rescaled, is approximately normal if the initial distribution is sufficiently regular. Since the family of functionally generated portfolios is convex, it can be viewed as an infinite dimensional constant-weighted family of portfolios.

\begin{problem}
Formulate and prove a version of Bernstein von-Mises Theorem in the setting of Theorem \ref{thm:main2}.
\end{problem}

\appendix
\section{}
The following lemmas are both standard results. Since we are unable to find suitable references, we will provide the proofs for completeness.

\begin{lemma} \label{lem:topology}
Let $X$ be a topological space and $Y$ be a subset of $X$ equipped with the subspace topology. If $A \subset Y$, then
\[
\partial_X A \subset \partial_Y A \cup \partial_X Y.
\]
\end{lemma}
\begin{proof}
We will argue by contradiction. Suppose $x \in \partial_X A$ and $x \notin \partial_Y A \cup \partial_X Y$.

By the definition of subspace topology and boundary, there exist neighborhoods $U_1$ and $U_2$ of $x$ in $X$ such that
\[
\text{(1) } U_1 \cap Y \subset A \quad \text{or} \quad \text{(2) } U_1 \cap Y \subset Y \setminus A,
\]
and
\[
\text{(i) } U_2 \subset Y \quad \text{or} \quad \text{(ii) } U_2 \subset X \setminus Y.
\]
We may replace $U_1$ and $U_2$ above by their intersection $U = U_1 \cap U_2$. Also, since $x \in \partial_X A$, $U$ intersects both $A$ and $X \setminus A$. We claim that the above statements are incompatible. We consider the following cases.

(1) and (i): Since $U \subset Y$ and $U \cap Y \subset A$, we have $U \subset A$. This contradicts the fact that $U$ intersects $X \setminus A$.

(2) and (i): We have $U \subset Y \setminus A$. But $A \subset Y$, so $U$ does not intersect $A$  and we have a contradiction.

(ii): If $U \cap Y = \emptyset$, then $U$ does not intersect $A$ which is a contradiction.
\end{proof}

\begin{lemma} \label{lem:weakconvergence}
Suppose ${\Bbb P}_t$ converges weakly to ${\Bbb P}$. Let $f: {\mathcal{S}} \rightarrow {\Bbb R}$ be bounded continuous and let $Y$ be a ${\Bbb P}$-continuity set in ${\mathcal{S}}$ with ${\Bbb P}(Y) > 0$. Then
\[
\lim_{t \rightarrow \infty} \int_Y f d {\Bbb P}_t = \int_Y f d{\Bbb P}.
\]
\end{lemma}
\begin{proof}
Consider the measures conditioned on $Y$:
\[
\widetilde{{\Bbb P}}_t(\cdot) = \frac{{\Bbb P}_t(\cdot \cap Y)}{{\Bbb P}_t(Y)}, \quad \widetilde{{\Bbb P}}(\cdot) = \frac{{\Bbb P}(\cdot \cap Y)}{{\Bbb P}(Y)}.
\]
Since ${\Bbb P}_t(Y) \rightarrow {\Bbb P}(Y) > 0$ as $A$ is a ${\Bbb P}$-continuity set, the measures $\widetilde{{\Bbb P}}_t$ are well defined for $t$ sufficiently large.

We claim that $\widetilde{{\Bbb P}}_t$ converges weakly to $\widetilde{\Bbb P}$. This implies the statement because $f$ is bounded continuous on $Y$ and
\[
\int_{{\mathcal{S}}} f d \widetilde{{\Bbb P}}_t = \frac{1}{{\Bbb P}_t(Y)} \int_Y f d {\Bbb P}_t \rightarrow \frac{1}{{\Bbb P}(Y)} \int_Y f d{\Bbb P} = \int_{{\mathcal{S}}} f d \widetilde{{\Bbb P}}.
\]

To prove the claim, it suffices by the Portmanteau theorem to show that $\widetilde{{\Bbb P}}_t(A) \rightarrow \widetilde{{\Bbb P}}(A)$ for all $A \subset Y$ with $\widetilde{{\Bbb P}}(\partial_Y A) = \frac{1}{{\Bbb P}(Y)} {\Bbb P} \left(\partial_Y A \cap Y\right) = 0$. Note that $\partial_Y A \subset Y$, so ${\Bbb P}\left(\partial_Y A\right) = 0$. By Lemma \ref{lem:topology}, we have $\partial_{{\mathcal{S}}} A \subset \partial_Y A \cup \partial_{{\mathcal{S}}} Y$, and so ${\Bbb P} \left(\partial_{{\mathcal{S}}} A\right) = 0$ as $Y$ is a ${\Bbb P}$-continuity set. Thus $A = A \cap Y$ is a ${\Bbb P}$-continuity set and we have ${\Bbb P}_t(A) \rightarrow {\Bbb P}(A)$. This completes the proof of the lemma.
\end{proof}

\section*{Acknowledgment}
The author thanks Soumik Pal for his suggestion to consider a market portfolio of portfolios and large deviations. Part of this research was done when the author was visiting UCSB in Spring 2015. He thanks the Department of Statistics and Applied Probability for its hospitality and Tomoyuki Ichiba for many helpful discussions. Some preliminary results of the paper were presented at the conference `Stochastic Portfolio Theory and related topics' at Columbia University in May 2015. He thanks the participants for their comments and suggestions. The author also thanks the anonymous referees and the editors for valuable comments. This research is supported partially by NSF grant DMS 1308340.

\bibliographystyle{plain}
\bibliography{infogeo}

\end{document}